\documentclass[10pt]{article}

\usepackage{amsmath,amssymb,amsthm,color,epsfig,mathrsfs}
\usepackage{amsbsy}

\newtheorem{theorem}{Theorem}

\newcounter{spslist}

\newcommand{\mat}[5]{ \renewcommand{\arraystretch}{#1}
                    \left[\! \begin{array}{cc}
                            #2 & #3 \\
                            #4 & #5 \end{array} \!\right] }

\newcounter{geqncount}
    {\refstepcounter{equation}%
     \setcounter{geqncount}{\value{equation}}%
     \setcounter{equation}{0}%
  }%
    {\setcounter{equation}{\value{geqncount}}}

\newcommand{\bi}{{\mathrm b}}

\newcommand{\ZZ}{\mathbb{Z}}
\newcommand{\RR}{\mathbb{R}}
\newcommand{\TT}{\mathbb{T}}
\newcommand{\CC}{\mathbb{C}}

\newcommand{\Hc}{\mathcal{H}}

\begin{document}

\bibliographystyle{plain}

\begin{center}
{\bf \Large  Stable defect states in the continuous spectrum \\ \vspace{0.5ex}  of bilayer graphene with magnetic field}
\end{center}

\vspace{0.2ex}

\begin{center}
{\scshape \large Stephen P\!. Shipman\footnote{shipman@lsu.edu, Department of Mathematics, LSU, Baton Rouge, LA 70808}
 \,and\, Jorge Villalobos\footnote{jvill38@lsu.edu, Department of Mathematics, LSU, Baton Rouge, LA 70808}
 } \\
\vspace{1ex}
{\itshape Department of Mathematics, Louisiana State University}
\end{center}

\vspace{3ex}
\centerline{\parbox{0.9\textwidth}{
{\bf Abstract.}\
In a tight-binding model of AA-stacked bilayer graphene, it is demonstrated that a bound defect state within the region of continuous spectrum can exist stably with respect to variations in the strength of a perpendicular magnetic field.  This is accomplished by creating a defect that is compatible with the interlayer coupling, thereby shielding the bound state from the effects of the continuous spectrum, which varies erratically in a pattern known as the Hofstadter butterfly.
}}

\vspace{3ex}
\noindent
\begin{mbox}
\end{mbox}
\vspace{3ex}

\hrule
\vspace{1.1ex}

\section{Introduction} 

When a magnetic field is imposed perpendicular to a sheet of graphene or other 2D periodic electronic structure, the energy spectrum is known to vary erratically with the strength of the magnetic field~$\phi$.  In an ideal model, when $\phi\!=\!2\pi \alpha$ with $\alpha\!=\!p/q$ rational, there is a band-gap spectral pattern, which becomes increasingly complex as $q\to\infty$ and deteriorates to a Cantor set when $\alpha$ is irrational.  There is order in this erratic variation, depicted by the Hofstadter butterfly diagram of spectrum versus magnetic field in Fig.~\ref{fig:butterfly}.  This phenomenon was noticed by Azbel~\cite{Azbel1964} and Hofstadter~\cite{Hofstadter1976} and has interesting physical ramifications~\cite{BeckerHanJitomZworski2023}.  The bulk spectrum has been proven to be of continuous type---absolute when the magnetic field is rational and singular (a Cantor set) when irrational~\cite{BeckerHanJitomirska2019}.

We are interested in the consequence of this erratic spectral behavior for defect states.  As the magnetic field strength is varied, a fixed energy $E$ within the energy range of the butterfly passes in and out of the spectrum.  Suppose that a local defect in the graphene layer for some magnetic strength $\phi$ engenders a bound state at an energy $E$ in a spectral gap within the butterfly picture.  As $\phi$ is varied, this resonant energy may enter a spectral band where the bound state will dissolve as it couples to extended Bloch states (for rational $\alpha$), or else the localization will be poor when $E$ is near a complicated part of a Cantor spectrum.  Such a bound state is therefore unstable under variation of the magnetic field.

The aim of our work is to use two interacting sheets of graphene to shield a defect state from the effects of the continuous bulk spectrum so that it persists stably as the magnetic field is varied.
By stability, we mean analytic dependence $E\!=\!E_\phi$ of the energy of a bound state on the strength $\phi$ of the magnetic field, with the bound state decaying at an exponential rate that is uniform over~$\phi$.
We accomplish this theoretically for AA-stacked bilayer graphene (when the sheets are aligned), when the defect is compatible with the interlayer coupling.  Construction of embedded eigenvalues for nonmagnetic multi-layered structures is demonstrated in~\cite{Shipman2014}, and even for AB-stacked graphene (when the sheets are relatively shifted) in \cite{FisherLiShipman2021}.  At the conclusion, we discuss the difficulties posed by AB-stacking with a magnetic field.

\begin{figure}[h]
\centerline{
\scalebox{0.40}{\includegraphics{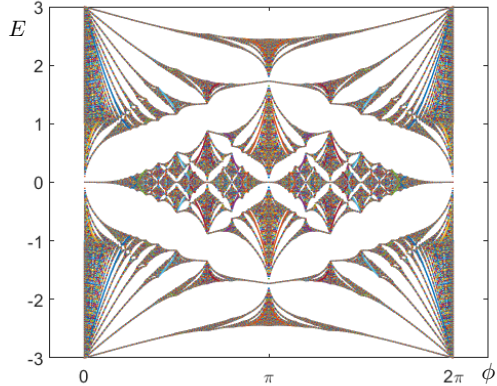}}
}
\caption{\small The ``Hofstadter butterfly" depicts the spectrum (vertical axis) of the magnetic tight-binding model for single-layer graphene as it depends on the magnetic field strength $\phi$ (horizontal axis).  When $\alpha=\phi/(2\pi)$ is irrational, it is a Cantor set, and when $\alpha=p/q$ it consists of $q$ bands.}
\label{fig:butterfly}
\end{figure}

In this study, we use a tight-binding model.  We first show how to obtain a bound state induced by a local defect in a single graphene sheet, with energy depending analytically on the magnetic field, as long as the energy remains sufficiently far from the butterfly region of spectrum.  We then couple two layers.  If the layers are aligned one directly over the other (AA-stacking), then the states of the bilayer graphene structure can be decomposed into two orthogonal spaces of hybrid states, each being an invariant space of the bilayer Hamiltonian.  Although the hybrid states in either of these invariant spaces are nonzero on both physical layers of graphene, their amplitudes on the two sheets are related, so each space of hybrid states acts as if it were the space of states of a single layer, but with shifted energy spectrum.  Each space of hybrid states contributes a Hofstadter butterfly to the $E$ {\it vs.} $\phi$ spectral picture for the bilayer graphene structure (Fig.~\ref{fig:Example}).

Then, we show that a local defect chosen to be compatible with the interlayer coupling will preserve the decoupling of hybrid states, so that scattering by the defect will not couple hybrid states from one invariant space with those of the other.  The compatibility is expressed as the commuting of the interlayer operator and the defect operator.  Then we apply the analysis for a single layer to obtain a stable bound state in one of the invariant spaces of hybrid states, lying outside its butterfly of continuous spectrum, but arranged so that it lies over the butterfly of the other hybrid space (Fig.~\ref{fig:Example}).  As the two hybrid spaces are invariant spaces of the Hamiltonian, and therefore dynamically decoupled, the bound states are oblivious to the presence of continuous spectrum varying erratically with the magnetic field.  The phenomena in this study apply to a Harper-type tight-binding model.  A careful discussion of tight-binding models with magnetic field and a rigorous comparison to continuum models can be found in~\cite{ShapiroWeinstein2022a}.

\section{Single-layer model} 

This section introduces the single-layer graphene model with magnetic field and shows how to construct defect states at energies in a spectral gap.  The results will be used in the following section to construct defect states for AA-stacked bilayer graphene that overlap with the continuous spectrum and are stable with respect to variation of the magnetic field.

In the simplest tight-binding model of single-layer graphene (Fig.~\ref{fig:defect}), there are two atoms per period cell and one orbital per atom, resulting in two degrees of freedom per period for the quantum mechanical state, or wave function, denoted by $u_1$ and~$u_2$.  When a perpendicular magnetic field with flux $\phi$ through each hexagon is imposed, the Hamiltonian $H_\phi$ is defined~by
\begin{equation}\label{H}
\begin{aligned}
    (H_\phi u)_1(n) &\;=\;u_2(n) + u_2(n+e_1) + e^{-in_1\phi}u_2(n+e_2) \\
    (H_\phi u)_2(n) &\;=\;u_1(n) + u_1(n-e_1) + e^{in_1\phi}u_1(n-e_2),
\end{aligned}
\end{equation}
in which $n=(n_1,n_2)$, $e_1=(1,0)$, and $e_2=(0,1)$.  
The phasors $e^{\pm in_1\phi}$ follow the Harper model of magnetism in tight binding~\cite{Harper1955}.  They also arise from a reduction of a quantum-graph model to a tight-binding one~\cite{BeckerHanJitomirska2019}.  We have chosen the Landau gauge for the magnetic potential.

Mathematically, $H_\phi$ is a self-adjoint operator on the Hilbert state space $\Hc = \ell^2(\ZZ^2,\CC^2)$, consisting of square-summable functions on the two-dimensional lattice $\ZZ^2$ with values in~$\CC^2$.  With reference to Fig.~\ref{fig:single}, let us describe it in more detail.  
Let $\mathbf{v}$ (red) and $\mathbf{w}$ (green) denote the two sites, or vertices, of a fixed period of the graphene structure.  These two sites are shifted by integer multiples of two vectors $\xi_1$ and~$\xi_2$ (Fig.~2) to create the hexagonal lattice.  
An element $u\in\Hc$ denotes a state as follows: For $n=(n_1,n_2)\in\ZZ^2$, $u(n)$ is a complex vector $( u_1(n), u_2(n) )$ in which $u_1(n)$ is the value of the state at the shifted red vertex at $\mathbf{v}+n_1\xi_1+n_2\xi_2$ and $u_2(n)$ is the value of the state at the shifted green vertex at $\mathbf{w}+n_1\xi_1+n_2\xi_2$.  The edges between vertices indicate nonzero matrix elements of the Hamiltonian operator $H_\phi$.

\begin{figure}[h]
\centerline{
\raisebox{8pt}{\scalebox{0.24}{\includegraphics{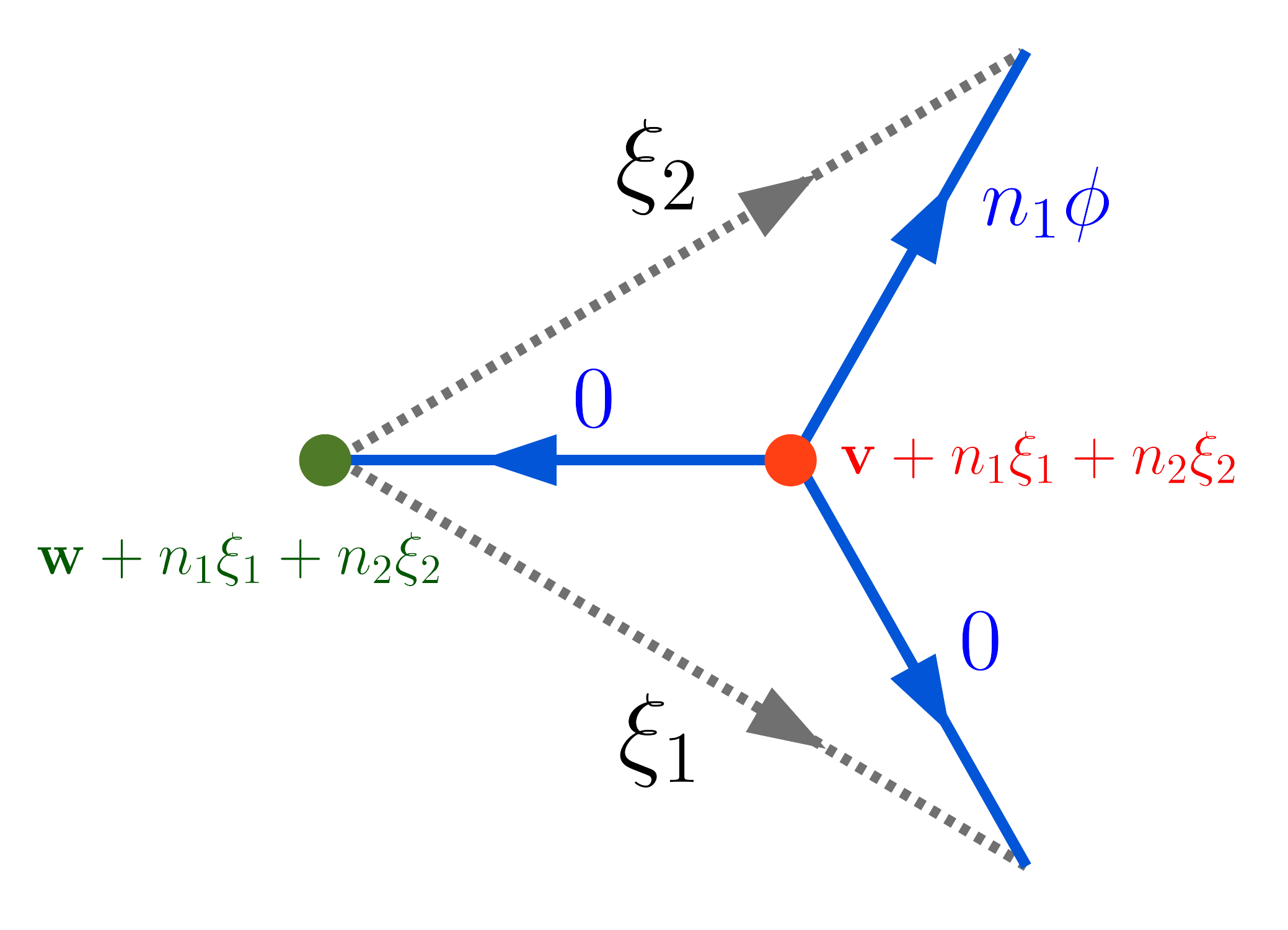}}}
\hspace*{2em}
\scalebox{0.4}{\includegraphics{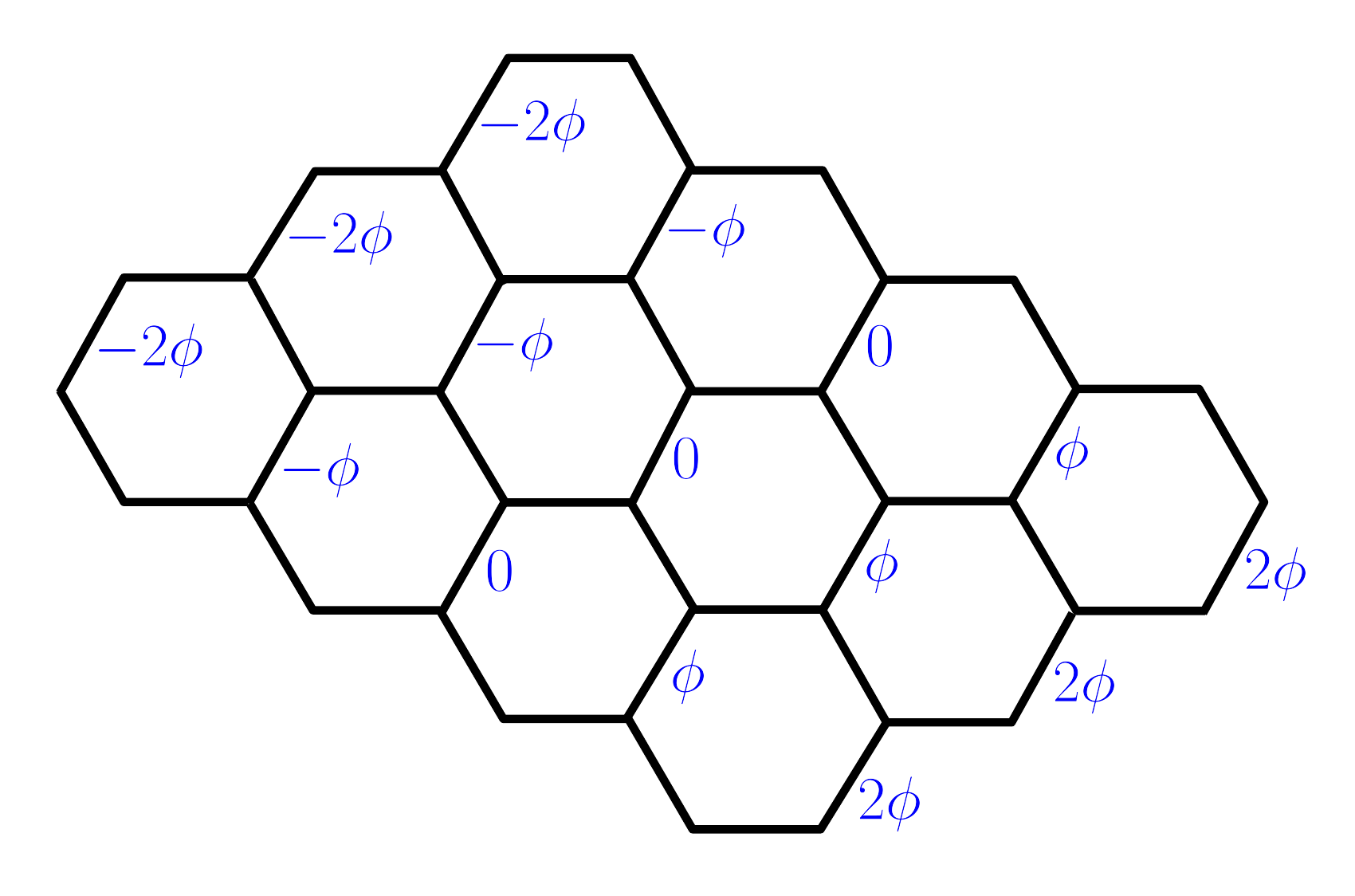}}
}
\caption{\small Left: The $n$-th shift of a fundamental domain of the graphene tight-binding model, where the red and green dots indicate the two orbitals associated with the two atoms in one period cell.  The elementary shifts by $e_1=(1,0)$ and $e_2=(0,1)$ are realized spatially as shifts by the vectors $\xi_1$ and $\xi_2$.  Right: A quasi-periodic magnetic potential on the edges that corresponds to a flux of $\phi$ through each hexagonal face, coming from a magnetic field directed perpendicular to the layer (Fig.~\ref{fig:defect}).}
\label{fig:single}
\end{figure}

\subsection{Preliminary analysis of the single layer}

The discrete Fourier transform converts the system into quasi-momentum space.  The components of the quasi-momentum $(k_1,k_2)$ are defined modulo~$2\pi$, and thus it is convenient to work with the ``Floquet multipliers" $(z_1,z_2)=(e^{ik_1},e^{ik_2})$ instead.  A state $u(n)$ is transformed to $\hat u(z)$ by 
\begin{equation}\label{Fourier}
  \hat u(z_1,z_2) \;=\; \sum_{n\in\ZZ^2} u(n_1,n_2)\, z_1^{-n_1}z_2^{-n_2},
  \qquad
  z=(z_1,z_2)\in\TT^2,
\end{equation}
in which $\TT^2=\{z=(z_1,z_2) : |z_1|=|z_2|=1 \}$ is the compact torus.
The Fourier transform is a unitary map from $\ell^2(\ZZ^2,\CC^2)$ to $L^2(\TT^2,\CC^2)$, with the normalized area (Haar measure) $dA(z)/(2\pi)^2=dk_1dk_2/(2\pi)^2$ on~$\TT^2$; and Fourier inversion gives
\begin{equation}
  u(n_1,n_2) \;=\; \frac{1}{4\pi^2} \int_{\TT^2} \hat u(z_1,z_2)\, z_1^{n_1}z_2^{n_2}\,dA(z).
\end{equation}

The Fourier transform converts $H_\phi$ into an operator $\hat H_\phi$ by the rule $\hat H_\phi \hat u=\left( H_\phi u \right)\hat{}$.
The utility of the Fourier transform is that $\hat u(z)$ is defined on the compact set~$\TT^2$, and additionally that shifts of $u(n)$ by $e_1$ and $e_2$ in the definition (\ref{H}) of $H_\phi$ are converted into multiplication of $\hat u(z)$ by $z_1$ and $z_2$ and that multiplication of $u(n)$ by $e^{in_1\phi}$ in the definition of $H_\phi$ is converted into a shift in momentum $k_1\mapsto k_1-\phi$, or $z_1\mapsto e^{-i\phi}z_1$, in the argument of~$\hat u(z)$.  Thus, $\hat H_\phi$ is given~by
\begin{equation}\label{Hhat}
\begin{aligned}
   (\hat H_\phi\hat u)_1(z_1,z_2) &\;=\; \hat u_2(z_1,z_2) + z_1\hat u_2(z_1,z_2) + z_2\hat u_2(e^{i\phi}z_1,z_2) \\
   (\hat H_\phi\hat u)_2(z_1,z_2) &\;=\; \hat u_1(z_1,z_2) + z_1^{-1}\hat u_1(z_1,z_2) + z_2^{-1}\hat u_1(e^{-i\phi}z_1,z_2).
\end{aligned}
\end{equation}
The operator $\hat H_\phi$ acts on $L^2(\TT^2,\CC^2)$, and it is unitarily equivalent to $H_\phi$ acting on~$\Hc$.
Using either (\ref{H}) or (\ref{Hhat}), one can see that $H_\phi$ is bounded,
\begin{equation*}
  \|H_\phi\| \leq 3.
\end{equation*}

The Fourier transform sum (\ref{Fourier}) extends analytically to a complex neighborhood of $\TT^2$ when $u(n)$ has exponential decay, according to the Paley-Wiener theorems.  Bounded analytic functions $\hat u(z)$ in the double $r$-annulus
\begin{equation*}
  A_r \;=\; \left\{ (z_1,z_2) : r^{-1} < |z_i|<r,\, i\in\{1,2\} \right\}
\end{equation*}
around the torus~$\TT^2$ correspond to functions $u(n)$ that decay exponentially as $|u(n_1,n_2)|<Cr^{-|n_1|-|n_2|}$.   For such functions, one has a uniform bound for $\hat H_\phi$.  More precisely, define the following norm on bounded analytic functions on~$A_r$:
\begin{equation*}
  \left\| \hat u \right\|_r \;=\; \max_{z\in A_r}\left| \hat u_1(z) \right| + \max_{z\in A_r}\left| \hat u_2(z) \right|.
\end{equation*}
Then, with respect to this norm, 
\begin{equation*}
  \big\| \hat H_\phi\hat u \big\|_r \,\leq\, M_r  \left\| \hat u \right\|_r
\end{equation*}
with $M_r\!=\! 1+2r $.

One can apply this bound to the equation $(H_\phi - E)u(n) \!=\! f(n)$ when $f(n)$ vanishes for sufficiently large~$n$, or, equivalently, when $\hat f(z)$ is a Laurent polynomial in~$(z_1,z_2)$, because in this case $\hat f(z)$ is bounded analytic in any annulus~$A_r$.  For energies $E$ such that $|E|\!>\!M_r$ the Fourier-transformed equation $(\hat H_\phi - E)\hat u(z) = \hat f(z)$ has a bounded analytic solution on~$A_r$ given by the Neumann series
\begin{equation*}
  \hat u(z) \;=\; -E^{-1}\sum_{\ell=0}^\infty \left[ (E^{-1}\hat H_\phi)^\ell\hat f \right]\!(z).
\end{equation*}
The response $u(n)$ is therefore exponentially decaying at least as~$Cr^{-|n|}$.
It is also analytic in the magnetic flux $\phi$ and energy~$E$ and can be expressed by Fourier inversion,
\begin{equation}
  u[\phi,E](n) \;=\; \frac{1}{4\pi^2} \int_{\TT^2} \hat u[\phi,E](z) z^n\,dA
\end{equation}
($z^n=z_1^{n_1}z_2^{n_2}$), and it is asymptotic to $-E^{-1}f(n)$ as $E\to\infty$.

The response $u[\phi,E](n)$ is actually analytic in $E$ in the resolvent set of~$H_\phi$, which, for each value of $\phi$, includes the set $\CC\setminus[-3,3]$ since~$\|H_\phi\|_\Hc\leq3$.  Its exponential decay degrades as $E$ approaches $[-3,3]$, as one must take $1+2r<|E|$ in the bound $|u[\phi,E](n)|<Cr^{-|n|}$.

\begin{figure}[h]
\centerline{
\scalebox{0.3}{\includegraphics{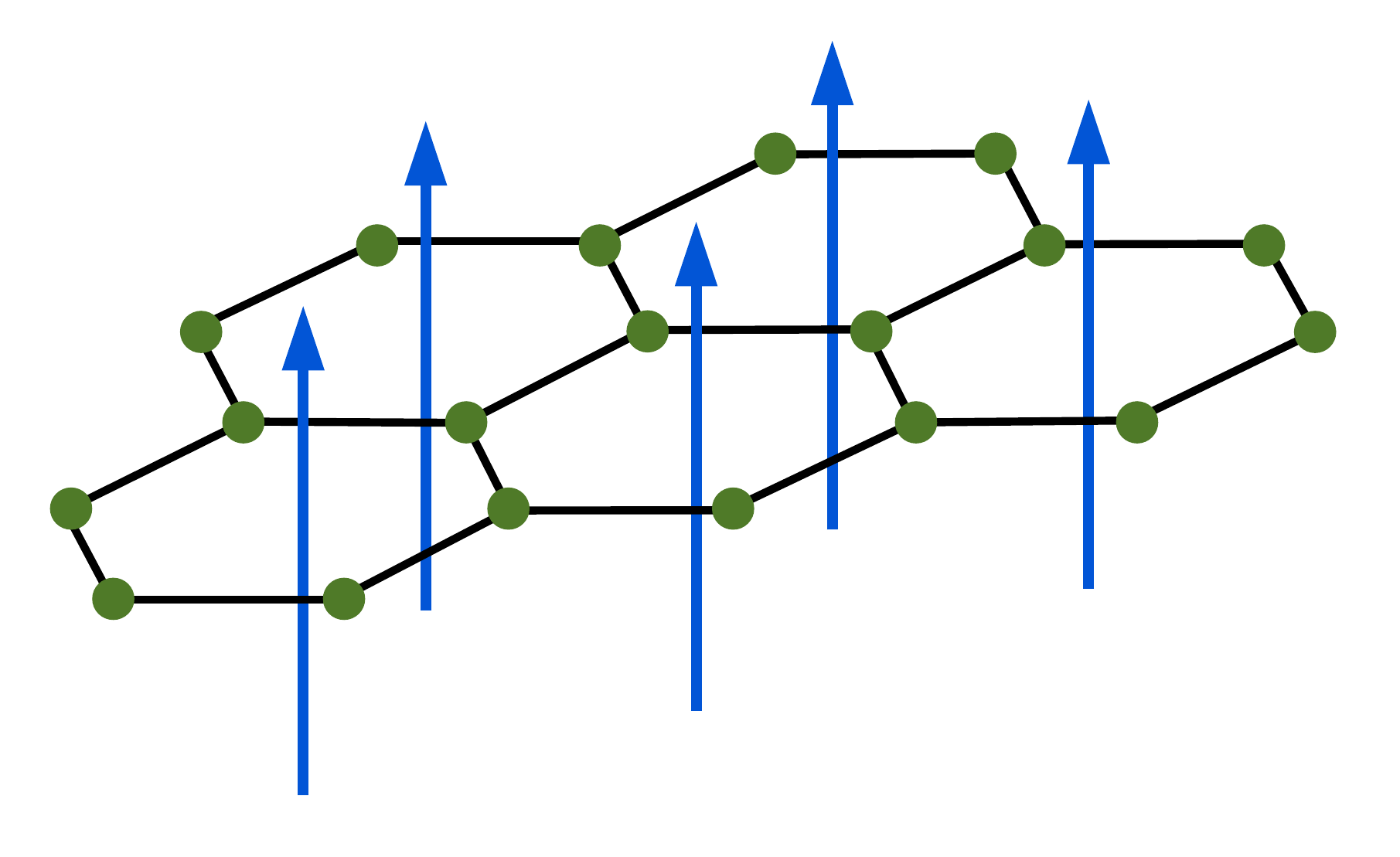}}
\hspace*{2em}
\scalebox{0.3}{\includegraphics{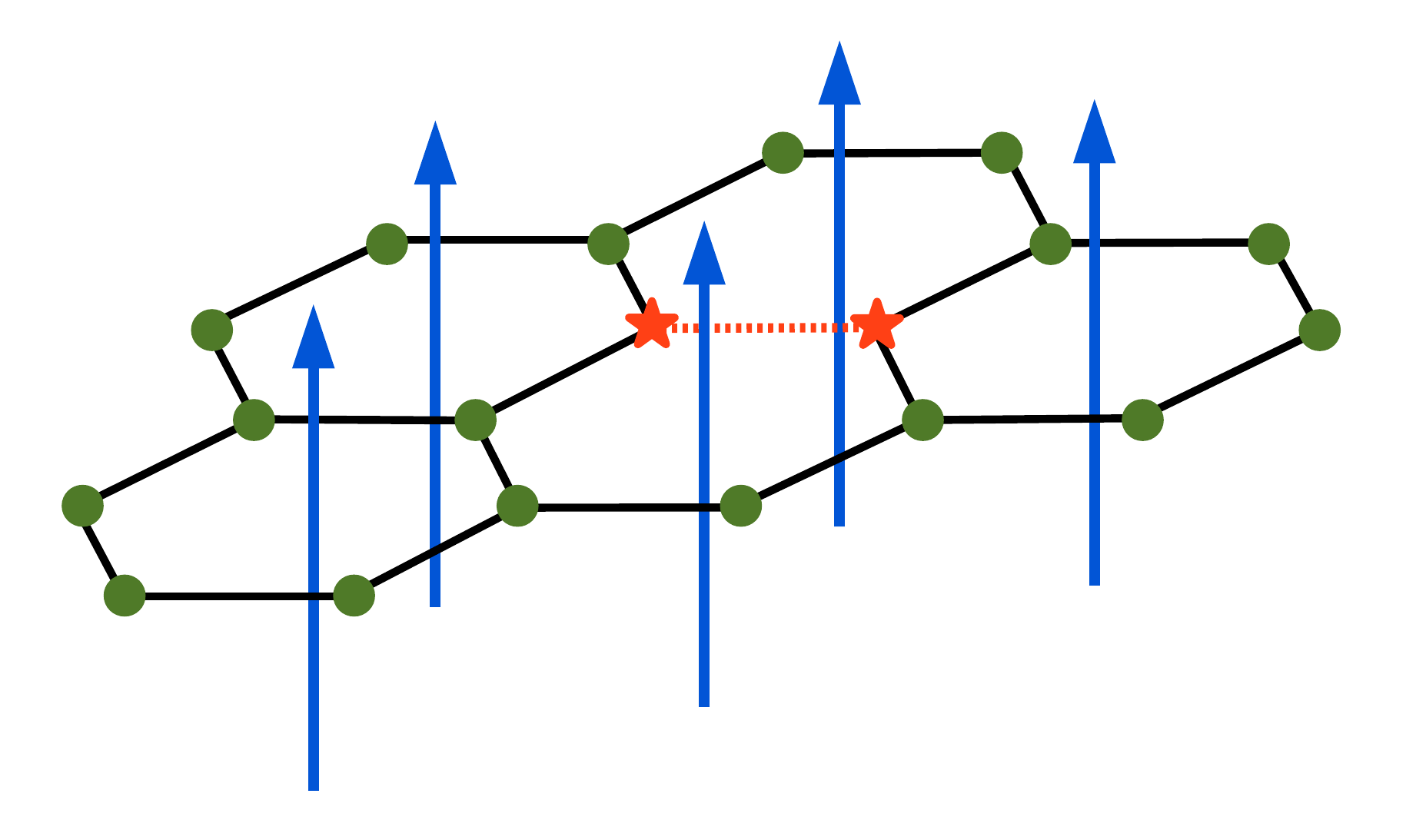}}
}
\caption{\small Left: A single graphene layer with a constant perpendicular magnetic field with flux $\phi$ through each hexagon.  Right:  A localized defect is created by adding a self-adjoint matrix to the two-by-two submatrix of the Hamiltonian corresponding to the interaction of two atoms.}
\label{fig:defect}
\end{figure}

\subsection{Defect states in a single layer}

A defect in the Hamiltonian $H_\phi$ is described by adding a local defect operator~$V$ to~$H_\phi$.  By a {\em local defect operator}, we mean a self-adjoint operator in $\Hc$ that satisfies $V\!=\!PVP$, where $P$ is the orthogonal projection onto the finite-dimensional space by~$\Hc_0$ consisting of states that vanish off of a finite set $S$ of sites of the hexagonal graphene structure.  We say that the defect operator {\em $V$ is localized at~$S$}.

\begin{theorem}\label{thm:one}
Let $\phi^0\in\RR$ and $E^0\in\RR$ such that $|E^0|>3$ be given, and let $v$ and $w$ be adjacent vertices of the hexagonal graph of the single-layer graphene model.  There exists a defect operator $V=PVP$ localized at the two-vertex set $\{v,w\}$; and real-analytic functions $\phi\mapsto E_\phi\in\RR$ and $\phi\mapsto u_\phi\in\Hc$ defined for $\phi$ in an open interval containing $\phi^0$, with $E_{\phi^0}\!=\!E^0$; and numbers $\gamma>0$ and $C>0$ such that
\begin{equation}\label{BIC2}
  (H_\phi + V)u_\phi \;=\; E_\phi\,u_\phi
  \qquad\text{and}\qquad
  |u_\phi(n)|<Ce^{-\gamma|n|}.
\end{equation}
\end{theorem}

\begin{proof}
Let $\phi^0\in\RR$ and $E^0\in\RR$ such that $|E^0|>3$ be given, and let $v$ and $w$ be adjacent vertices of the hexagonal graph of the single-layer graphene model.  Let $f$ be a function of the vertices that equals $1$ on $v$ and vanishes otherwise.  Since $E^0$ is in the resolvent set of $H_\phi$, there exists a unique $u\in\Hc$ such that
\begin{equation}\label{HEuf}
  (H_\phi - E^0)u = f.
\end{equation}

Let $w'$ be either of the two other vertices adjacent to~$v$.  We will show that $|u(w)|=|u(w')|$.  Let $R$ denote the rotation of the hexagonal graph about the vertex $v$ that rotates $w$ into~$w'$, and let $R$ also denote the induced action on $\Hc$, that is $Ru(x)=u(R^{-1}x)$ for all vertices~$x$.  The hexagonal graph associated with the operator $RH_\phi R^{-1}$ is a rotation of that of $H_\phi$---all of the coefficients still have unit modulus, but the phases change; particularly, it has magnetic flux $\phi$ through all hexagons.  This means that the rotation merely effects a change of magnetic potential, and thus there is a gauge transformation $G$ such that $RH_\phi R^{-1}=GH_\phi G^{-1}$; see \cite{LiebLoss1993a}.  The gauge transformation is a unitary multiplication operator on $\Hc$, meaning that, for each state $f$ and vertex $x$, $Gf(x)=e^{i\psi_x}f(x)$ for some real number~$\psi_x$.  We obtain $\Upsilon(H_\phi-E^0)\Upsilon^{-1}=H_\phi-E^0$ with $\Upsilon=G^{-1}R$.  This, together with (\ref{HEuf}) implies $(H_\phi-E^0)\Upsilon u=\Upsilon f=e^{-i\psi_v}f$.  The uniqueness of the solution to (\ref{HEuf}) yields $\Upsilon u=e^{-i\psi_v}u$, and therefore
\begin{equation}
  Ru \;=\; e^{-i\psi_v}Gu.
\end{equation}
This provides the relation
\begin{equation}
  u(w) = u(R^{-1}w') = Ru(w') = e^{-i\psi_v}Gu(w') = e^{i(\psi_{w'}-\psi_v)}u(w'),
\end{equation}
and thus $|u(w)|=|u(w')|$, as desired.

By the definition (\ref{H}) of~$H_\phi$, since $f(v)\!\not=\!0$, $u$ cannot vanish simultaneously at $v$ and all of the three vertices adjacent to~$v$.  Thus either $u(v)\!\not=\!0$ or $u(w')\!\not=\!0$ for some vertex $w'$ adjacent to~$v$.  In the latter case, $|u(w)|=|u(w')|\not=0$.  In either case, the vector $[ u(v), u(w) ]^t$ is nonzero and $[ f(v), f(w) ]^t=[ 1,0 ]^t$.  Let $P$ be the orthogonal projection to the two-dimensional space $\Hc_0$ of states that are supported on the two vertices $v$ and~$w$, and let $V=PVP$ be a local defect operator such that its action in $\Hc_0$ is represented by a Hermitian matrix $\tilde V$ such that $-\tilde V[ u(v), u(w) ]^t=[ f(v), f(w) ]^t$.  This is equivalent to $-Vu=f$, and we obtain
\begin{equation}
  (H_\phi +V - E^0)u = 0.
\end{equation}
Since $V$ is a compact perturbation of $H_\phi$ and $E^0$ is in the resolvent set of $H_\phi$, $E^0$ is an isolated eigenvalue of the self-adjoint operator $H_\phi+V$ in $\Hc$.  Analytic perturbation theory guarantees, for $\phi$ in an interval about~$\phi^0$, a local analytic function $E_\phi$ with $E_{\phi^0}=E^0$ and states $u_\phi\in\Hc$ depending analytically on~$\phi$, such that $(H_\phi + V)u_\phi \!=\! E_\phi\,u_\phi$.  Since $|E^0|>3$, one can restrict $\phi$ to an interval about $\phi^0$ such that $|E_\phi|\!>1+2r\!>\!3$ for some $r>1$.  As discussed above, one obtains the exponential decay with~$r=e^\gamma$.
\end{proof}

We believe that, generically, the defect can be localized at one vertex only, as all that is required for this is that, when a nonzero forcing $f$ is applied at one vertex $v$, the response $u$ is nonzero at~$v$.  One can prove this when $\phi=0$, as $\hat H_0$ is a multiplication operator and the response can be computed explicitly.

\section{Bilayer structures and spectrally embedded bound states} 

Consider now two identical layers in AA-stacking, that is, one placed one directly above the other as shown in Fig.~\ref{fig:Double}.  The state space is now $\Hc\oplus\Hc$, and a general state $(u_1,u_2)$ consists of a component $u_1$ residing on one layer and a component $u_2$ on the other layer.  When representing operators on this state space, it is convenient to write it isomorphically as $\CC^2\otimes\Hc$.  In this notation, a state $(u_1,u_2)$ can be written as $(1,0)\otimes u_1+(0,1)\otimes u_2$, or, for example, as $(1,1)\otimes v_1+(1,-1)\otimes v_2$, where $v_1=(u_1+u_2)/2$ and $v_2=(u_1-u_2)/2$ are the even and odd components of $(u_1,u_2)$.  Block-form operators on $\Hc\oplus\Hc$ can be written in terms of tensor products of operators on $\CC^2\otimes\Hc$.  If $K=(\kappa_{ij})_{i,j\in\{1,2\}}$ is a $2\times2$ matrix and $L$ is an operator on $\Hc$, then $K\otimes L$ written in block form is
\begin{equation}
  K\otimes L \;=\; \mat{1.1}{\kappa_{11}L}{\kappa_{12}L}{\kappa_{21}L}{\kappa_{22}L},
\end{equation}
and its action is given by $(K\otimes L)(\xi\otimes u) = K\xi\otimes Lu$.
The matrix $K$ is viewed as an interlayer operator, and the operator $L$ as an intralayer operator.

A bi-layer AA-stacked graphene Hamiltonian with perpendicular magnetic field strength $\phi$~is
\begin{equation}\label{Hbi}
  H^\bi_\phi \;=\; I\otimes H_\phi \,+\, K\otimes I. 
\end{equation}
The defective Hamiltonian is $H^\bi_\phi + D$, where the defect operator has the form
\begin{equation}\label{D}
   D \;=\; M\otimes V,
\end{equation}
with $V$ being an intralayer local defect operator on~$\Hc$, as described in the previous section.
In the first term of (\ref{Hbi}), $I$ is the $2\times2$ identity matrix, so this term realizes two uncoupled copies of the single-layer Hamiltonian.  The second term represents spatially uniform coupling of the two layers: $I$ is the identity operator in $\Hc$, and the self-adjoint matrix $K$ represents the interlayer coupling strengths.  One could consider a more general coupling operator $K\!\otimes\!L$, where $L$ is a periodic self-adjoint operator in $\Hc$, but we will take $L\!=\!I$, which corresponds to the simplest AA-stacked model, where each site of one layer interacts only with the corresponding site of the other layer as depicted by vertical connecting edges in Fig.~\ref{fig:Double}.   The off-diagonal entries of $K$ represent reciprocal orbital coupling, and the difference between the diagonal entries achieves a bias through the potential of a perpendicular electric field (this is called ``gating").  The defect term of (\ref{D}) has self-adjoint interlayer coupling matrix~$M$ and is localized by~$V$.

\begin{figure}[h]
\centerline{
\scalebox{0.27}{\includegraphics{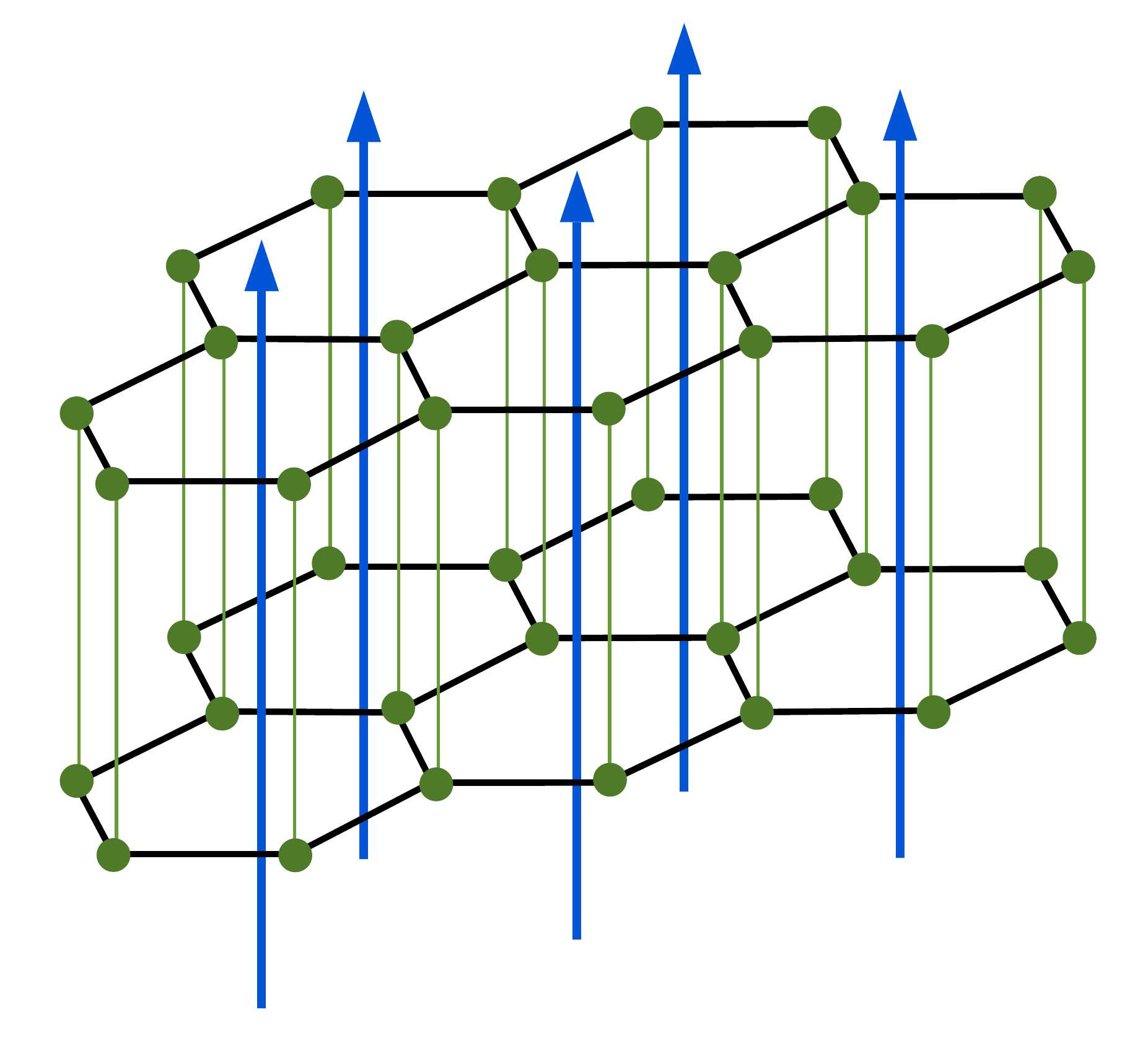}}
\hspace*{2em}
\scalebox{0.27}{\includegraphics{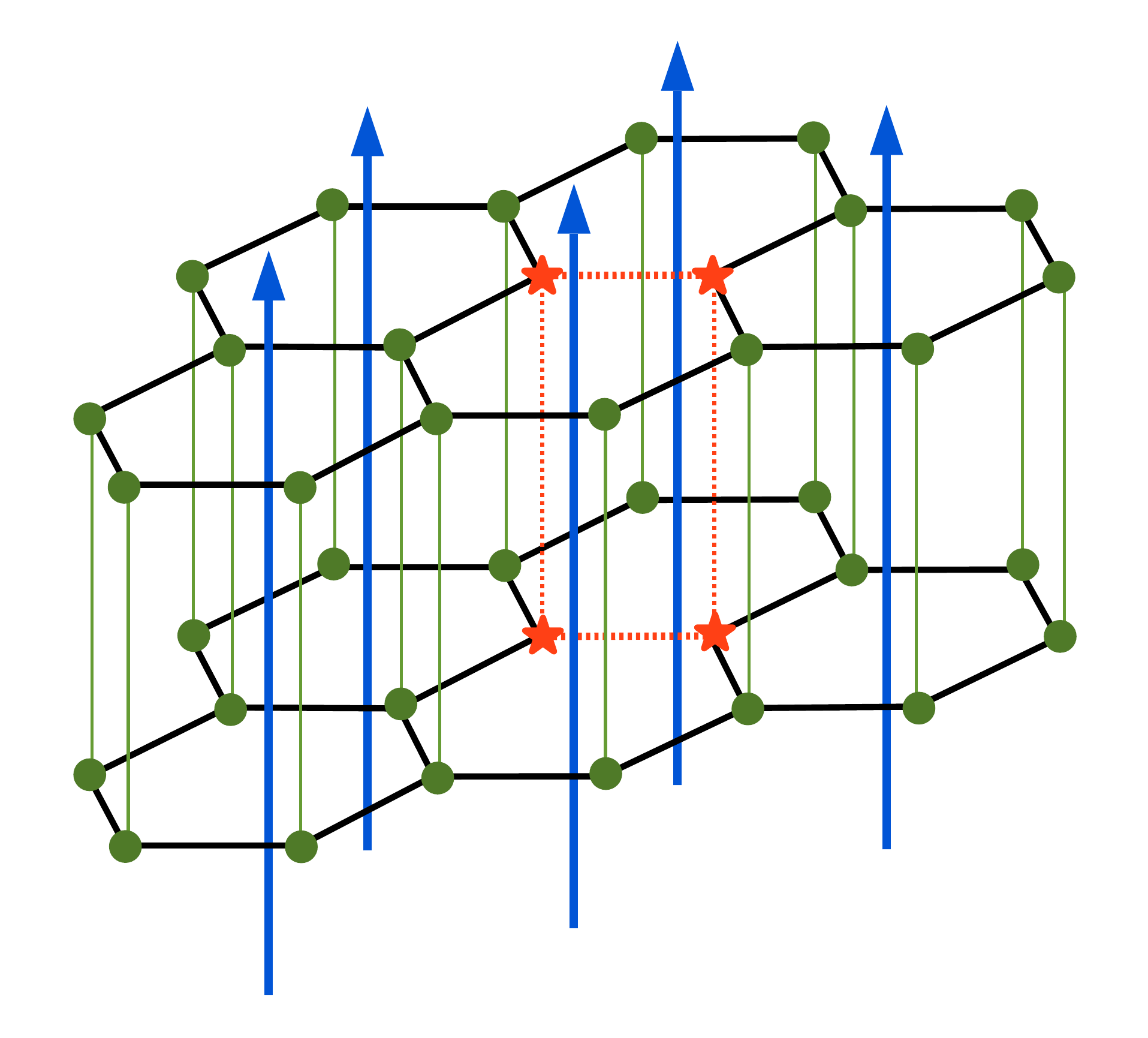}}
}
\caption{\small Left: AA-stacked graphene layers with a perpendicular magnetic field.  Right: A local defect in bilayer graphene with a perpendicular magnetic field.}
\label{fig:Double}
\end{figure}

Compatibility of the defect with the coupling, mentioned in the introduction, means that $K$ and $M$ commute, so that they have common eigenvectors $\xi^{(1)}$ and~$\xi^{(2)}$:
\begin{equation}
  K\xi^{(i)} = \kappa_i \xi^{(i)},
  \quad
  M\xi^{(i)} = \mu_i \xi^{(i)},
  \qquad i\in\{1,2\}.
\end{equation}
Thus, $H^\bi_\phi+D$ is decoupled by the decomposition of states into $H^\bi_\phi$-invariant spaces of hybrid states,
\begin{equation}
  \CC^2\otimes\Hc \;=\; \Hc^{(1)}\oplus\Hc^{(2)},
\end{equation}
with each hybrid space $\Hc^{(i)}$ being isomorphic to $\Hc$,
\begin{equation}
  \Hc^{(i)} \;=\; \left\{ \xi^{(i)}\!\otimes\!u \;\big|\;  u \in\Hc \right\}.
\end{equation}
The action of $H^\bi_\phi+D$ on $\Hc^{(i)}$ is
\begin{equation}
  (H^\bi_\phi+D)(\xi^{(i)}\!\otimes\!u) \;=\; \xi^{(i)} \otimes \left[ H_\phi +\kappa_i I + \mu_i V \right]u.
\end{equation}
Let us denote the restrictions of $H^\bi_\phi+D$ to the subspaces $\Hc^{(i)}$ by
\begin{equation}
  H^{(i)}_\phi \;=\; (H^\bi_\phi+D) |_{\Hc^{(i)}} \;=\; H_\phi +\kappa_i I + \mu_i V,
\end{equation}
in which the second equality incurs a slight relaxation of notation, as we omit to write the prefactor~``$\xi^{(i)}\otimes$".

Note that, when $K=a I + b\sigma_1$ in the Pauli matrix notation, the eigenvectors are $\xi^{(1)}\!=\!\langle 1,1 \rangle$ and $\xi^{(2)}\!=\!\langle 1,-1 \rangle$, and the hybrid states are the even and odd ones with respect to reflection about the parallel plane between the two layers.

\smallskip
The hybrid-state operators $H^{(i)}_\phi$ are spectrally shifted versions of the single-layer defective operator $H_\phi+V$, with relative energy shift $\kappa_2\!-\!\kappa_1$ and proportional defects $V\!=\!\mu_i D$.
Because the defect contributes a compact perturbation of the underlying periodic operator, the continuous spectrum remains unchanged, and we have
\begin{equation}\label{spectrum}
  \sigma_c(H^{(1)}_\phi)=\sigma_c(H_\phi)+\kappa_1,
  \qquad
  \sigma_c(H^{(2)}_\phi)=\sigma_c(H_\phi)+\kappa_2.
\end{equation}
The spectrum of $H^\bi_\phi+D$ is the union of the spectra of $H^{(1)}_\phi$ and $H^{(2)}_\phi$.  Because of (\ref{spectrum}), each of these includes a shifted copy of the continuous spectrum of the non-defective single-layer operator $H_\phi$, as depicted in Fig.~\ref{fig:Example}.  We are interested in energies outside the spectrum of one of the operators, say~$H^{(1)}_\phi$, but within the region of possible spectrum (depending on $\phi$) of the other operator~$H^{(2)}_\phi$.  Assuming $0<\kappa_2-\kappa_1<6$, this includes the interval
\begin{equation}
  E \,\in\, [\kappa_1,\kappa_2] + 3.
\end{equation}
One can apply Theorem~\ref{thm:one} to $H^{(1)}_\phi$, choosing $E$ to lie in the continuous spectrum of $H^{(2)}_{\phi}$ to obtain a bound state of $H^\bi_\phi+D$ within its continuous spectrum.
 Theorem~\ref{thm:two} states that, for any magnetic field strength $\phi$ and any energy $E$ outside the interval $[-3,3]+\kappa_1$, which contains $\sigma_c(H^{(1)}_\phi)$, a defect operator $D$ can be constructed so that the defective AA-stacked bilayer graphene model $H^\bi_\phi+D$ admits an exponentially decaying bound state at $E$, and that the energy and bound state vary analytically with perturbations of $\phi$.  {\em The point is that the eigenvalue can be chosen to lie within the continuous spectrum of $H^{(2)}_\phi$ and therefore also within the continuous spectrum of~$H^\bi_\phi+D$.}

\begin{theorem}\label{thm:two}
Let $K$ and $M$ be commuting Hermitian $2\times2$ matrices with common eigenvectors $\xi^{(i)}$ ($i\in\{1,2\}$); and let $K$ have corresponding eigenvalues $\kappa_i$ ($i\in\{1,2\}$) and let $M$ have corresponding eigenvalues $\mu_i$ ($i\in\{1,2\}$) with $\mu_1\not=0$.
Let $\phi^0\in\RR$ and $E^0\in\RR$ with $E^0\not\in[-3,3]+\kappa_1$ be given, and let $v$ and $w$ be adjacent vertices of the single-layer graphene structure.

There exists an intralayer defect operator $V=PVP$ localized at the two-vertex set $\{v,w\}$; and real-analytic functions $\phi\mapsto E_\phi\in\RR$ and $\phi\mapsto\underline u_\phi\in\CC^2\otimes\Hc$ defined for $\phi$ in an open interval containing~$\phi^0$, with $E_{\phi^0}=E^0$; and numbers $\gamma>0$ and $C>0$ such that
\begin{equation}\label{solution}
  (H^\bi_\phi+D)\underline u_\phi\;=\;E_\phi\underline u_\phi
  \qquad\text{and}\qquad
  |\underline u_\phi(n)|<Ce^{-\gamma|n|},
\end{equation}
in which $H^\bi_\phi = I\otimes H_\phi \,+\, K\otimes I$ and $D=M\otimes V$ and $H_\phi$ is defined by~(\ref{H}).

If $\kappa_1\not=\kappa_2$, then $\phi^0$ and $E^0$ can be chosen such that $E^0$ is contained in the continuous spectrum of $H^\bi_{\phi^0}+D$.
\end{theorem}

\begin{proof}
  Let $2\times2$ matrices $K$ and $M$, numbers $\phi\in\RR$ and $E^0\in\RR$, and adjacent vertices $v$ and $w$ of the single graphene layer be given as in the hypotheses of the theorem.  Since $|E^0-\kappa_1|>3$, Theorem~\ref{thm:one} provides a defect operator $W$ supported on $\{v,w\}$, an analytic relation $\tilde E_\phi$ and states $u_\phi$, for $\phi$ in a neighborhood of $\phi^0$, such that $\tilde E_{\phi^0}=E^0-\kappa_1$ and
\begin{equation}\label{one}
  ( H_\phi + W )u_\phi \;=\; \tilde E_\phi\, u_\phi.
\end{equation}
Set $E_\phi = \tilde E_\phi+\kappa_1$ and $V=W/\mu_1$.  Then $E_{\phi^0}=E^0$, and (\ref{one}) implies
$[ H_\phi + \kappa_1 I + \mu_1 V ] u_\phi = E_\phi\,u_\phi$, that~is,
\begin{equation}\label{two}
  H^{(1)}_\phi u_\phi \;=\; E_\phi\,u_\phi,
\end{equation}
and by putting $\underline u_\phi = \xi^{(1)}\otimes u_\phi$, one obtains
\begin{equation}
   (H^\bi_\phi+D)\underline u_\phi\;=\;E_\phi\underline u_\phi
\end{equation}
The exponential bound is also provided by Theorem~\ref{thm:one}.

When $\kappa_1\not=\kappa_2$ and $\phi=0$ one has $\sigma_{c}(H^{(1)}_\phi)=[-3,3]+\kappa_1$ and $\sigma_{c}(H^{(2)}_\phi)=[-3,3]+\kappa_2$, so $E^0$ can be chosen such that $E^0\not\in[-3,3]+\kappa_1$ but $E^0\in\sigma_{c}(H^{(2)}_\phi)$.  The last statement of the theorem then follows, considering that $\sigma_c(H^\bi_\phi+D)=\sigma_{c}(H^{(1)}_\phi)\cup\sigma_{c}(H^{(2)}_\phi)$.
\end{proof}

\begin{figure}[h]
\centerline{
\scalebox{0.46}{\includegraphics{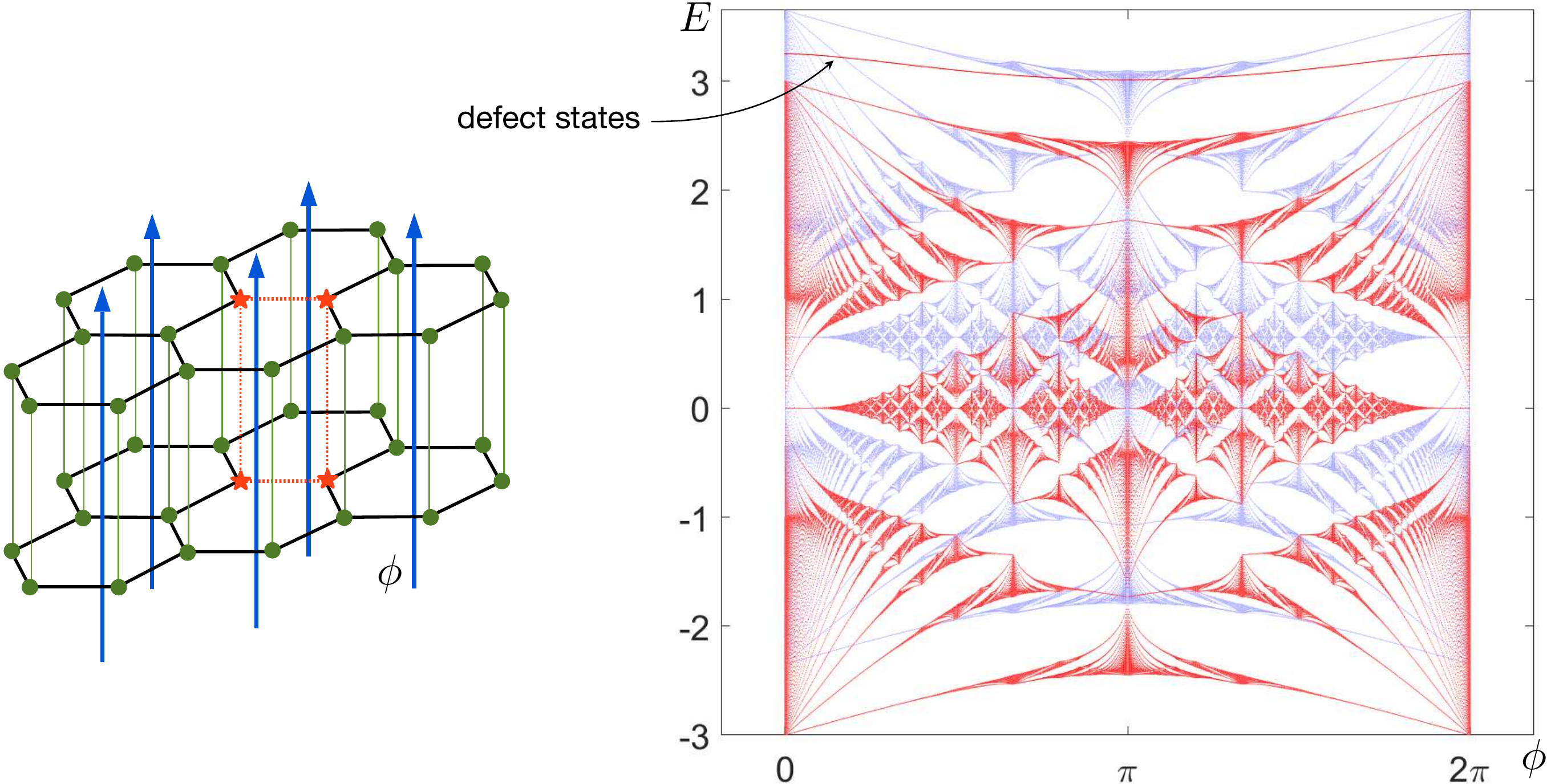}}
}
\caption{\small The continuous spectrum of a tight-binging model for bilayer AA-stacked graphene with a magnetic field exhibits two copies of the Hofstadter butterfly, which depicts energy $E$ {\itshape vs.\!} magnetic field strength~$\phi$ (right).  The two copies correspond to two decoupled spaces of layer-hybrid states, $\Hc^{(1)}$ and $\Hc^{(2)}$, with spectral shift $\kappa_2\!-\!\kappa_1=0.65$.  The red (lower) butterfly depicts the continuous spectrum for the space~$\Hc^{(1)}$, and the blue (upper) butterfly depicts the continuous spectrum for the~$\Hc^{(2)}$.
A defect is placed at four vertices, indicated with red stars (left).  The curve $E_\phi$ (right top red) shows the relation between $E$ and $\phi$ of a defect state in the space $\Hc^{(1)}$ of hybrid states. 
The curve passes through the continuous spectrum of the other hybrid space~$\Hc^{(2)}$.} 
\label{fig:Example}
\end{figure}

We compute numerically a relation $E\!=\!E_\phi$ in an example where the defect potential is localized at four vertices of the bilayer graphene model (Fig.~\ref{fig:Example}), so that the defect in each of the reduced operators $H_\phi^{(i)}$ is localized at two adjacent vertices, as in Theorem~\ref{thm:one}.  We compute a defect-state relation $E=E_\phi$ for the operator $H_\phi^{(1)}$, which lies outside its butterfly region $[-3,3]+\kappa_1$ of spectrum.  Then, one can independently adjust the coupling $K$ such that $\kappa_1$ and $\kappa_2$ will shift the butterfly region $[-3,3]+\kappa_2$ of $H_\phi^{(2)}$ so that it overlaps the defect-state relation.  Figure~\ref{fig:Example} shows how the energy $E$ of the defect state varies periodically with the magnetic field~$\phi$.

We wish to mention an interesting phenomenon related to defect states in the continuum.  In the proof of Theorem~\ref{thm:one}, creating a defect state started with the equation $(H_\phi-E)u=f$ with $f\in\Hc$, and $u$ being exponentially decaying.  In the context of Theorem~\ref{thm:two}, one obtains the equation
\begin{equation*}
  (I\otimes H_\phi \,+\, K\otimes I - E)\underline u \;=\; \underline f.
\end{equation*}
This equation can hold for energies $E$ within the region of continuous spectrum of 
the non-defective bi-layer graphene model $I\otimes H_\phi \,+\, K\otimes I$.  This means that a localized forcing $\underline f$ results in a response $\underline u$ that is exponentially decaying even for energies within the continuum.

\section{Conclusion and Discussion}  

The aim of this work is to create a defect in a planar graphene structure that engenders a bound state which persists under the introduction of a perpendicular magnetic field, and furthermore whose energy lies within the region of continuous spectrum described by the Hofstadter butterfly.  The energy of the defect state should vary smoothly with the strength of the magnetic field and have a controlled rate of exponential decay.  Such a defect state would act stably with respect to variations in the magnetic field, unaffected by the presence of an erratically changing continuous spectrum.  Conventional wisdom says that continuous spectrum around the energy of a defect state should destroy its localization.  {\itshape In this work, we have demonstrated that a bound state can be shielded from the effects of continuous spectrum when two layers are AA-stacked.}

The analysis exploits a reduction of the bilayer Hamiltonian to invariant spaces of hybrid states, plus compatibility of the coupling and defect Hamiltonians.  When there is no gating provided by a perpendicular electric field, the two hybrid spaces consist of even and odd states with respect to reflection about the center plane between the layers.  Compatibility of the coupling and defect ensures that the hybrid spaces remain non-interacting when the defect is introduced, and this in turn allows a bound state outside the continuous spectrum of one hybrid space to coincide with but not interact with the continuous spectrum of the other hybrid space.

Analysis by reducibility of the Hamiltonian to invariant spaces is particular to AA-stacking, which is not the energy-minimizing layering.  Energy is lower in the stable AB-stacked configuration, in which one atom of a layer sits above the center of a hexagon of the other layer, where there is no atom, as discussed, for example, in \cite{DaiXiangSrolovitz2016a} and \cite{CazeauxLuskinMassatt2020a}.  Whether AB-stacked graphene admits defect states within the spectral continuum that vary stably with a magnetic field remains an open question. 

Interestingly, in the absence of a magnetic field, AB- and certain more general multi-layer periodic models admit defect states embedded in the continuous spectrum---even when there are no Hamiltonian-invariant hybrid spaces.  This is due to a deeper, more general, kind of reducibility that is algebraic in nature, namely the reducibility of the Fermi surface at all energies, which allows the construction of such states; this is described in~\cite{FisherLiShipman2021}.  The Fermi surface for periodic graph operators is an algebraic set, and it is known that its reducibility is necessary for the creation of spectrally embedded defect states~\cite{KuchmentVainberg2006} (except for certain peculiar states supported at a finite number of sites, which are particular to graph models).  The introduction of a magnetic field, however, destroys the Fermi surface of the underlying periodic graph because the magnetic potential in the Hamiltonian is generically not periodic, but quasi-periodic.

The analysis via reduction of the Hamiltonian to invariant spaces of hybrid states remains intact for general periodic tight-binding configurations of the single layer, not just hexagonal structures, when the layers are aligned.  It is also extensible to multiple aligned layers, and the number of invariant hybrid state spaces is equal to the number of layers~\cite{Shipman2014}.

It appears that multiple layers are typically necessary for defect states in the continuum.  This is not quite a rigorous mathematical statement; and deciding what counts as a single layer is not so well defined, although it is clear in specific situations, such as graphene.  In the absence of a magnetic potential, one can make a few rigorous statements.  As mentioned above, the reducibility of the Fermi surface is needed for a local defect to produce a spectrally embedded bound state.  This means that at a given energy $E$, the dispersion function $D(z_1,z_2,E)$ can be factored as a Laurent polynomial in $(z_1,z_2)$, where $D(z_1,z_2,E)$ is the function whose zero set gives momentum-energy pairs for which the periodic operator admits a Floquet-Bloch mode.  The Fermi surface for single-layer graphene is irreducible.  It is also known to be irreducible for the discrete (tight-binding) Laplace operator plus a periodic potential~\cite{GiesekerKnorrerTrubowitz1993,FillmanLiuMatos2022} and for planar tight-binding models with two vertices per period \cite{LiShipman2020}, and all these models certainly count as a single layer.  Additionally, known models with reducible Fermi surface involve multiple layers attached in particular ways~\cite{FisherLiShipman2021}, and the layers are used deliberately to construct the irreducible components of the Fermi surface.  Construction of spectrally embedded defect states  is then reduced to algebra through the Fourier transform, and this works for very general stacking of graphene.  As mentioned above, because a magnetic field destroys the periodicity of the Hamiltonian, the algebraic Fermi surface and Fourier analysis is no longer directly applicable, and thus the investigation of the existence of defect states stable under a varying magnetic field, particularly for AB-stacked graphene, will require new techniques.

\vspace{5ex}

\noindent{\bfseries Acknowledgement.}
This work is dedicated to the memory of Prof. Hermann Flaschka.
The material is based upon work supported by the National Science Foundation under Grant No.\,DMS-2206037. 
J.\,Villalobos thanks the University of Costa Rica for its support of his studies at Louisiana State University.


\end{document}